\documentclass[10pt,conference]{IEEEtran}

\usepackage{times}
\usepackage{amsmath,dsfont}
\usepackage{amssymb,amsthm}
\usepackage{epsfig,verbatim}
\usepackage{mathrsfs}
\usepackage{verbatim}
\usepackage{syntonly}
\usepackage{graphicx}
\usepackage{epstopdf}
\usepackage{multirow}
\usepackage{subfigure}

\theoremstyle{remark} \newtheorem{theorem}{Theorem}
\theoremstyle{remark} 
\newtheorem{proposition}{Proposition}
\newtheorem*{propositionA}{Proposition}

\theoremstyle{remark} \newtheorem{remark}{Remark}

\newcommand{\eps}{\epsilon}

\newcommand{\styp}{A_{\epsilon}^{*(n)}}

\newcommand{\mX}{\mathcal{X}}

\newcommand{\mY}{\mathcal{Y}}
\newcommand{\mW}{\mathcal{W}}
\newcommand{\mV}{\mathcal{V}}
\newcommand{\mS}{\mathcal{S}}

\newcommand{\yvec}{\mathbf{y}}
\newcommand{\xvec}{\mathbf{x}}

\newcommand{\svec}{\mathbf{s}}

\newcommand{\avec}{\mathbf{a}}

\newcommand{\wvec}{\mathbf{w}}

\newcommand{\tsvec}{\tilde{\svec}}
\newcommand{\hsvec}{\hat{\svec}}

\newcommand{\msg}{u}
\newcommand{\msgBig}{\MakeUppercase{\msg}}
\newcommand{\msgCal}{\mathcal{\msgBig}}
\newcommand{\tmsg}{\tilde{\msg}}
\newcommand{\hmsg}{\hat{\msg}}

\long\def\symbolfootnote[#1]#2{\begingroup\def\thefootnote{\fnsymbol{footnote}}\footnote[#1]{#2}\endgroup}

\newcommand{\thmlabel}[1]{\label{thm:#1}}
\newcommand{\thmref}[1]{\ref{thm:#1}}
\newcommand{\Thmref}[1]{Thm.~\thmref{#1}}

 \IEEEoverridecommandlockouts
\title{Joint Source-Channel Coding for the Multiple-Access Relay Channel
\thanks{
\noindent
This work was partially supported by the European Commission's Marie Curie IRG Fellowship 
PIRG05-GA-2009-246657 under the Seventh Framework Programme. Deniz G\"{u}nd\"{u}z is partially supported by the European Commission's Marie Curie Fellowship IRG256410, and by the Spanish Government under project TEC2010-17816 (JUNTOS).
}
}

\author{
\IEEEauthorblockN{Yonathan Murin, Ron Dabora}
\authorblockA{Department of Electrical and Computer Engineering \\
Ben-Gurion University, Israel\\
Email: moriny@bgu.ac.il, ron@ee.bgu.ac.il}

\vspace{-0.75cm}

\and
\IEEEauthorblockN{Deniz G\"und\"uz}
\authorblockA{Centre Tecnologic de Telecomunicacions  \\
de Catalunya (CTTC), Barcelona, Spain \\
Email: deniz.gunduz@cttc.es}

\vspace{-0.75cm}

}

\vspace{-0.6cm}

\begin{document}

\maketitle

\thispagestyle{empty} \pagestyle{empty}


\vspace{-0.2cm}

\begin{abstract}
    Reliable transmission of arbitrarily correlated sources over multiple-access relay channels (MARCs) and multiple-access broadcast relay channels (MABRCs) is considered. In MARCs, only the destination is interested in a reconstruction of the sources, while in MABRCs, both the relay and the destination want to reconstruct the sources. We allow an arbitrary correlation among the sources at the transmitters, and let both the relay and the destination have side information that are correlated with the sources. 
    Two joint source-channel coding schemes are presented and the corresponding sets of sufficient conditions for reliable communication are derived. The proposed schemes use a combination of the correlation preserving mapping (CPM) technique with Slepian-Wolf (SW) source coding: the first scheme uses CPM for encoding information to the relay and SW source coding for encoding information to the destination; while the second scheme uses SW source coding for encoding information to the relay and CPM for encoding information to the destination. 
\end{abstract}

\vspace{-0.3cm}

\section{Introduction}

\vspace{-0.1cm}

The multiple-access relay channel (MARC) models a network in which several users communicate with a single destination, with the help of a relay \cite{Kramer:2005}. Examples of such a network include sensor and ad-hoc networks in which an intermediate relay node is introduced to assist the communication from the source terminals to the destination. The MARC is a fundamental multi-terminal channel model that generalizes both the multiple-access channel (MAC) and the relay channel models, and has received a lot of attention in the recent years. If the relay terminal is also required to decode the source messages, the model is called the multiple-access broadcast relay channel (MABRC).

While in \cite{Kramer:2005},\cite{Sankar:07} MARCs with independent sources at the terminals are considered, in \cite{Murin:ISWCS11}, \cite{Murin:IT11} we allow arbitrary correlation among the sources to be transmitted to the destination in a lossless fashion. 
We also let the relay and the destination have side information that are correlated with the two sources.
In \cite{Murin:ISWCS11} we address the problem of determining whether a pair of sources can be losslessly transmitted to the destination with a given number of channel uses per source sample, using statistically independent source code and channel code.

In \cite{Shannon:48} Shannon showed that a source can be reliably transmitted over a memoryless point-to-point (PtP) channel, if and only if its entropy is less than the channel capacity. Hence, a simple comparison of the rates of the optimal source code and the optimal channel code for the respective source and channel, suffices to conclude whether reliable communication is feasible. This is called the {\it separation theorem}. 
However, the optimality of separation does not generalize to multiuser networks \cite{Cover:80}, \cite{GunduzErkip:09}, and, in general the source and channel codes must be jointly designed for every particular combination of source and channel, for optimal performance.

In this paper we study {\em source-channel coding for the transmission of correlated sources over MARCs}.
We note that while the capacity region of the MAC (which is a special case of the MARC) is known for independent messages, the optimal joint source-channel code for the case of correlated sources is not known in general \cite{Cover:80}. 
%
%
Single-letter sufficient conditions for communicating discrete, arbitrarily correlated sources over a MAC are derived in \cite{Cover:80}. These conditions were later shown by Dueck in \cite{Dueck:81}, to be sufficient but not necessary. This gives an indication on the complexity of the problem studied in the present work.

The main technique used in \cite{Cover:80} is the {\em correlation preserving mapping (CPM)} in which the channel codewords are correlated with the source sequences. Since the source sequences are correlated with each other, CPM leads to correlation between the channel codewords. 
The CPM technique of \cite{Cover:80} is extended to source coding with side information for the MAC in \cite{Ahlswede:83} and to broadcast channels with correlated sources in \cite{HanCosta:87}. 
Transmission of arbitrarily correlated sources over interference channels (ICs) is studied in \cite{LiuChen:2011}, in which Liu and Chen apply the CPM technique to ICs.
Lossless transmission over a relay channel with correlated side information is studied in \cite{ErkipGunduz:07} and \cite{ElGamalCioffi:07}. In \cite{ErkipGunduz:07} a decode-and-forward (DF) based achievability scheme is proposed and it is shown that separation is optimal for physically degraded relay channels with degraded side information, as well as for cooperative relay-broadcast channels with arbitrary side information. 
Necessary and sufficient conditions for reliable transmission of a source over a relay channel, when side information is available at the receiver or at the relay, are established in \cite{ElGamalCioffi:07}.

\vspace{-0.2cm}

\subsection*{Main Contributions}

\vspace{-0.15cm}

In this paper we first demonstrate the suboptimality of separate source and channel encoding for the MARC by considering the transmission of correlated sources over a discrete memoryless (DM) semi-orthogonal MARC in which the relay-destination link is orthogonal to the channel from the sources to the relay and the destination.

Next, we propose two DF-based joint source-channel achievability schemes for MARCs and MABRCs.
 Both proposed schemes use a combination of SW source coding and the CPM technique. While in the first scheme CPM is used for encoding information to the relay and SW source coding is used for encoding information to the destination; in the second scheme SW source coding is used for encoding information to the relay and CPM is used for encoding information to the destination. A comparison of the conditions of the two schemes reveals a tradeoff: while the relay feasibility conditions of the former are looser, the destination feasibility conditions of the latter are looser. These are {\em the first joint source-channel achievability schemes, proposed for a multiuser network with a relay, which take advantage of the CPM technique.}

The rest of this paper is organized as follows: in Section \ref{sec:NotationModel} we introduce the system model and notations. In Section \ref{subsec:exampleSepSubOpt} we demonstrate the suboptimality of separate encoding for the MARC. In Section \ref{sec:JointAchiev} we present two achievability schemes for DM MARCs and MABRCs with correlated sources and side information, and derive their corresponding sets of feasibility conditions. We discuss the results in Section \ref{sec:jointDiscussion}, and conclude the paper in Section \ref{sec:conclusions}.

%

\vspace{-0.1cm}

\section{Notations and System Model} \label{sec:NotationModel}
\vspace{-0.08cm}

In the following we denote random variables with upper case letters, e.g. $X$, and their realizations with lower case letters, e.g. $x$. 
A discrete random variable $X$ takes values in a set $\mX$. 
We use $p_X(x)\equiv p(x)$ to denote the probability mass function (p.m.f.) of a discrete RV $X$ on $\mX$.
We denote vectors with boldface letters, e.g. $\xvec$; 
the $i$'th element of a vector $\xvec$ is denoted by $x_i$, and we use $\xvec_i^j$ where $i<j$ to denote $(x_i, x_{i+1},...,x_{j-1},x_j)$; $x^j$ is a short form notation for $x_1^j$.
We use $\styp(X)$ to denote the set of $\eps$-strongly typical sequences w.r.t. the p.m.f $p_X(x)$ on $\mX$, as defined in \cite[Ch. 13.6]{cover-thomas:it-book}. When referring to a typical set we may omit the random variables from the notation, when these variables are clear from the context. 
The empty set is denoted by $\phi$.


The MARC consists of two transmitters (sources), a receiver (destination) and a relay.
Transmitter $i$ observes to the source sequence $S_i^n$, for $i=1,2$.
The receiver is interested in the lossless reconstruction of both source sequences observed by the two transmitters.
The objective of the relay is to help the receiver decode these sequences.
Let $W_3^n$ and $W^n$, denote the side information at the relay and at the receiver respectively. The side information sequences are correlated with the source sequences.
For the MABRC both the receiver and the relay are interested in a lossless reconstruction of both source sequences.
Figure \ref{fig:MABRCsideInfo} depicts the MABRC with side information setup. 
\begin{figure}[h]
    \centering
    \scalebox{0.39}{\includegraphics{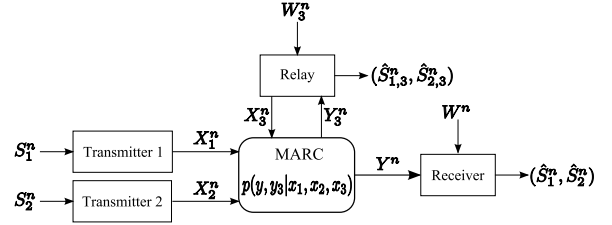}}
    \vspace{-0.2cm}
    \caption{Multiple-access broadcast relay channel with correlated side information. 
    $(\hat{S}^n_{1,3}, \hat{S}^n_{2,3})$ are the reconstructions of $(S^n_{1}, S^n_{2})$ at the relay, and $(\hat{S}^n_{1}, \hat{S}^n_{2})$ are the reconstructions at the destination.}
    \label{fig:MABRCsideInfo}
    \vspace{-0.6cm}
\end{figure}

The sources and the side information sequences, $\{ S_{1,k},S_{2,k},W_{k},W_{3,k} \}_{k=1}^{n}$, are
arbitrarily correlated according to a joint distribution $p(s_1,s_2,w,w_3)$ over a
finite alphabet $\mS_1 \times \mS_2 \times \mW \times \mW_3$, and independent across different sample indices $k$.
All nodes know this joint distribution.

For transmission, a discrete memoryless MARC with inputs $X_1, X_2, X_3$ over finite input alphabets $\mX_1,\mX_2,\mX_3$, and outputs $Y, Y_3$ over finite output alphabets $\mY,\mY_3$, is available.
The MARC is memoryless in the sense
\vspace{-0.15cm}
\begin{equation}
	p(y_{k},y_{3,k}|y^{k-1} \mspace{-1.3mu} ,y_3^{k-1} \mspace{-1.3mu} ,x_1^k,x_2^k,x_3^k) \mspace{-3mu} = \mspace{-3mu} p(y_{k},y_{3,k}|x_{1,k},x_{2,k},x_{3,k}). \nonumber
\label{eq:MARCchanDist}
\end{equation}

\vspace{-0.15cm}

A source-channel code for MABRCs with correlated side information consists of two encoding functions at the transmitters:
$f_i^{(n)} : \mS_i^n \mapsto \mX_i^n, i=1,2$, a decoding function at the destination,
$g^{(n)}: \mY^n \times \mW^n \mapsto \mS_1^n \times \mS_2^n$, and a decoding function at the relay, $g_3^{(n)}: \mY_3^n \times \mW_3^n \mapsto \mS_1^n \times \mS_2^n$.
Finally, there is a causal encoding function at the relay, $x_{3,k} = f_{3,k}^{(n)}(y_{3,1}^{k-1},w_{3,1}^n), 1 \leq k \leq n$.
Note that in the MARC scenario the decoding function $g_3^{(n)}$ does not exist.
Let $\hat{S}_i^n$ and $\hat{S}_{i,3}^n$ denote the reconstruction of $S_i^n, i=1,2,$ at the receiver and at the relay respectively.
The average probability of error of a source-channel code for the 
MABRC is defined as $P_e^{(n)} \triangleq  \Pr\big((\hat{S}_1^n,\hat{S}_2^) \neq (S_1^n,S_2^n) \mbox{ or } (\hat{S}_{1,3}^n,\hat{S}_{2,3}^n) \neq (S_1^n,S_2^n) \big)$. For the MARC the definition is similar except that the decoding error event at the relay is omitted.
%
The sources $(S_1,S_2)$ can be \emph{reliably transmitted} over the MABRC with side information if there exists a sequence of source-channel codes such that $P_e^{(n)} \rightarrow 0$ as $n \rightarrow \infty$.
The same definition applies to MARCs.

Before presenting the new joint source-channel coding schemes, we first motivate this work by demonstrating the suboptimality of separate encoding for the MARC.

\vspace{-0.15cm}
\section{Suboptimality of Separation for DM MARCs} 
 \label{subsec:exampleSepSubOpt}

 
Consider the transmission of arbitrarily correlated sources $S_1$ and $S_2$ over a DM semi-orthogonal MARC (SOMARC), in which the relay-destination link is orthogonal to the channel from the sources to the relay and to the destination.
The SOMARC is characterized by the joint distribution $p(y_R,y_S,y_3|x_1,x_2,x_3)=p(y_R|x_3)p(y_S,y_3|x_1,x_2)$, where $Y_R$ and $Y_S$ are the channel outpus at the destination. The SOMARC is depicted in Figure \ref{fig:SOMARC}. 
In the following we present a scenario (sources and a channel) in which joint source-channel coding strictly outperforms separate source-channel coding.

We begin with an outer bound on the sum-capacity of the SOMARC. This is characterized in Proposition \ref{prop:ExmpCahnnel}.
\begin{figure}[h]
    \vspace{-0.2cm}
    \centering
    \scalebox{0.40}{\includegraphics{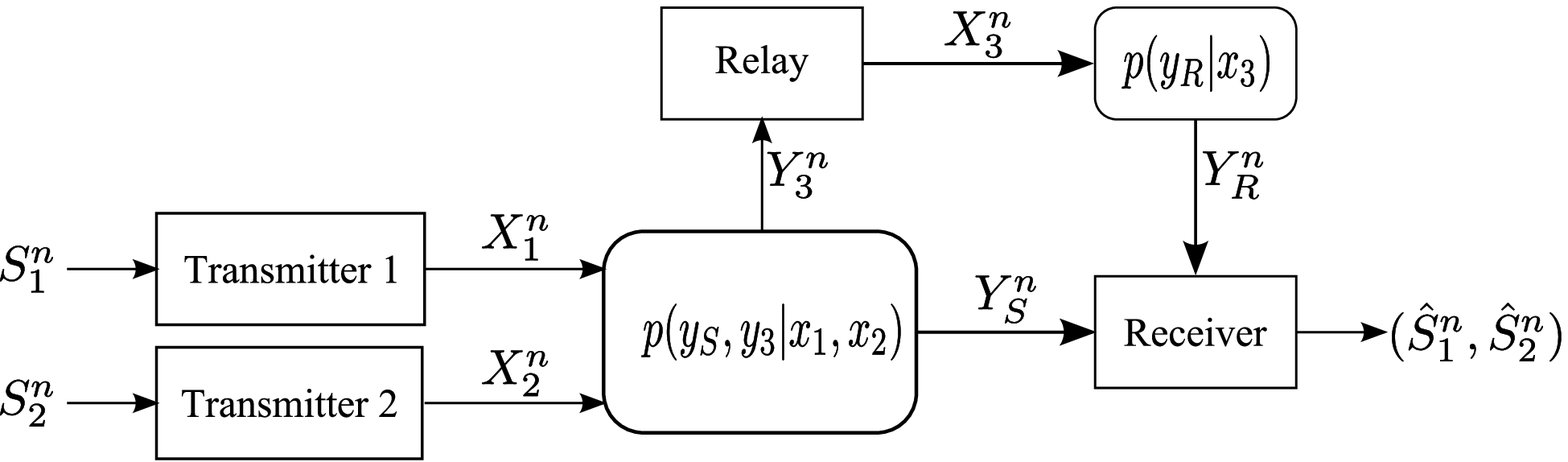}}
    \vspace{-0.15cm}
    \caption{Semi-orthogonal multiple-access relay channel.}
    \label{fig:SOMARC}
    \vspace{-0.5cm}
\end{figure}

\begin{proposition} \label{prop:ExmpCahnnel}
    The sum-capacity of the SOMARC, $R_1+R_2$, is upper bounded by
    \vspace{-0.2cm}
    \begin{align}
        R_1+R_2 \leq \max_{p(x_1)p(x_2)p(x_3)} & \min \big\{ I(X_1, X_2; Y_3,Y_S), \nonumber \\
        & I(X_3; Y_R) + I(X_1,X_2;Y_S) \big\}.
    \label{eq:SOMARC_Bound}
    \end{align}
    \vspace{-0.5cm}
\end{proposition}

\begin{proof}[$\mspace{-42mu}$ Proof]
    Detailed proof is provided in \cite[Subsection VI-A]{Murin:IT11}.
\end{proof}

		 Next, consider a SOMARC defined by
  \vspace{-0.25cm}
\begin{eqnarray}
    \mX_1 &=& \mX_2 = \mX_3 = \mY_3 = \mY_R = \{0,1\}, \quad \mY_S = \{0,1,2\}, \nonumber \\
    Y_R &=& X_3, \quad Y_3 = X_1 \oplus X_2, \quad Y_S = X_1 + X_2.
\label{eq:SOMARC_channel}
\end{eqnarray}
	\vspace{-0.6cm}

Additionally, consider the sources $(S_1,S_2) \in \{0,1\} \times \{0,1\}$ with the joint distribution $p(s_1,s_2)=\frac{1}{3}$ for $(s_1,s_2) \in \{ (0,0), (0,1), (1,1) \}$, and zero otherwise. Then,  $H(S_1,S_2)= \log_2 3 = 1.58$ $\text{bits} / \text{sample}$. 
For the channel defined in \eqref{eq:SOMARC_channel} the mutual information expression $I(X_1, X_2; Y_3, Y_S)$ reduces to $I(X_1, X_2; Y_S)$. This is because $I(X_1, X_2; Y_3,Y_S)=I(X_1, X_2; Y_S)+I(X_1, X_2; Y_3 |Y_S)$, and $Y_3$ is a deterministic function of $Y_S$. Therefore, 
	\vspace{-0.2cm}
\begin{align}
        R_1+R_2 & \leq 
				\max_{p(x_1)p(x_2)} I(X_1, X_2; Y_S) \nonumber \\
        & =  1.5 \text{ bits per channel use}.
    \end{align}
    \vspace{-0.6cm}

\noindent Hence, we have $H(S_1,S_2) > I(X_1,X_2;Y_S)$, for any $p(x_1)p(x_2)$. We conclude that it is not possible to send the sources reliably to the destination by using a separation-based source and channel codes. However, by choosing $X_1=S_1$ and $X_2=S_2$ a zero error probability is achieved. This example shows that {\em separate source and channel coding is, in general, suboptimal for sending arbitrarily correlated sources over MARCs}.

\vspace{-0.1cm}

\section{Joint Source-Channel Coding for Discrete Memoryless MARCs and MABRCs} \label{sec:JointAchiev}

\vspace{-0.08cm}

In this section we present two sets of sufficient conditions for reliable transmission of correlated sources over DM MARCs and MABRCs with side information.
Both achievability schemes use a combination of SW source coding, the CPM technique, and the DF scheme with successive decoding at the relay and backward decoding at the destination \cite{Kramer:2005}. We shall refer to SW source coding as separate source and channel coding (SSCC).
The achievability schemes differ in the way the source codes are combined. 
In the first scheme (\Thmref{thm:jointCond}) SSCC is used for encoding information to the destination while CPM is used for encoding information to the relay. In the second scheme (\Thmref{thm:jointCondFlip}), CPM is used for encoding information to the destination while SSCC is used for encoding information to the relay. 
\begin{theorem}
    \thmlabel{thm:jointCond}
    A source pair $(S_1^n,S_2^n)$ can be reliably transmitted over a DM MARC with relay and receiver side information as defined in Section \ref{sec:NotationModel} if,
    \vspace{-0.10cm}
    \begin{subequations} \label{bnd:Joint}
    \begin{eqnarray}
        H(S_1|S_2,W_3) &<& I(X_1;Y_3|S_2, X_2, V_1, X_3, W_3) \label{bnd:Joint_rly_S1} \\
        H(S_2|S_1,W_3) &<& I(X_2;Y_3|S_1, X_1, V_2, X_3, W_3) \label{bnd:Joint_rly_S2} \\
        H(S_1,S_2|W_3) &<& I(X_1,X_2;Y_3|V_1, V_2, X_3, W_3) \label{bnd:Joint_rly_S1S2} \\
        H(S_1|S_2,W) &<& I(X_1,X_3;Y|S_1, X_2, V_2) \label{bnd:Joint_dst_S1} \\
        H(S_2|S_1,W) &<& I(X_2,X_3;Y|S_2, X_1, V_1) \label{bnd:Joint_dst_S2} \\
        H(S_1,S_2|W) &<& I(X_1,X_2,X_3;Y|S_1,S_2), \label{bnd:Joint_dst_S1S2}
    \end{eqnarray}
    \end{subequations}
    \vspace{-0.5cm}

    \noindent for a joint distribution that factors as
    \vspace{-0.2cm}
    \begin{align}
        & p(s_1,s_2,w_3,w)p(v_1)p(x_1|s_1,v_1) \times \nonumber \\ 	
        & \qquad	p(v_2)p(x_2|s_2,v_2)p(x_3|v_1,v_2)p(y_3,y|x_1,x_2,x_3).
    \label{eq:JntJointDist}
    \end{align}
    \vspace{-0.6cm}
\end{theorem}

\begin{proof}[$\mspace{-42mu}$ Proof]
    See \cite[Subsection VI.D, Appendix C]{Murin:IT11}.
\end{proof}   

\begin{theorem}
    \thmlabel{thm:jointCondFlip}
    A source pair $(S_1^n,S_2^n)$ can be reliably transmitted over a DM MARC with relay and receiver side information as defined in Section \ref{sec:NotationModel} if,
    \vspace{-0.2cm}
    \begin{subequations} \label{bnd:JointFlip}
    \begin{eqnarray}
        H(S_1|S_2,W_3) &<& I(X_1;Y_3|S_1, X_2, X_3) \label{bnd:JointFlip_rly_S1} \\
        H(S_2|S_1,W_3) &<& I(X_2;Y_3|S_2, X_1, X_3) \label{bnd:JointFlip_rly_S2} \\
        H(S_1,S_2|W_3) &<& I(X_1,X_2;Y_3|S_1, S_2, X_3 ) \label{bnd:JointFlip_rly_S1S2} \\
        H(S_1|S_2,W) &<& I(X_1,X_3;Y|S_2, X_2, W) \label{bnd:JointFlip_dst_S1} \\
        H(S_2|S_1,W) &<& I(X_2,X_3;Y|S_1, X_1, W) \label{bnd:JointFlip_dst_S2} \\
        H(S_1,S_2|W) &<& I(X_1,X_2,X_3;Y| W), \label{bnd:JointFlip_dst_S1S2}
    \end{eqnarray}
    \end{subequations}
    \vspace{-0.55cm}

    \noindent for a joint distribution that factors as
    \vspace{-0.2cm}
    \begin{align}
        &p(s_1,s_2,w_3,w)p(x_1|s_1)p(x_2|s_2) \times \nonumber \\
        & \qquad \quad p(x_3|s_1,s_2)p(y_3,y|x_1,x_2,x_3).
    \label{eq:JntFlipJointDist}
    \end{align}
		\vspace{-0.8cm}
\end{theorem}

\begin{proof}[$\mspace{-42mu}$ Proof]

\subsubsection{Codebook construction}
For $i=1,2$, assign every $\svec_i \in \mS_i^n$ to one of $2^{nR_i}$ bins independently according to a uniform distribution on $\msgCal_i \triangleq \{1,2,\dots,2^{nR_i}\}$. Denote this assignment by $f_i, i=1,2$.

For $i=1,2$, for each pair $(\msg_i, \svec_i), \msg_i \in \msgCal_i, \svec_i \in \mS_i^n$, generate one $n$-length codeword $\xvec_i(\msg_i, \svec_i)$, by choosing the letters $x_{i,k}(\msg_i, \svec_i)$ independently with distribution $p_{X_i|S_i}(x_{i,k}|s_{i,k})$ for all $1 \leq k \leq n$.
Finally, generate one length-$n$ relay codeword $\xvec_3(\svec_1,\svec_2)$ for each pair $(\svec_1,\svec_2) \in \mS_1^n \times \mS_2^n$ by choosing $x_{3,k}(\svec_1,\svec_2)$ independently with distribution $p_{X_3|S_1,S_2}(x_{3,k}|s_{1,k},s_{2,k})$ for all $1 \leq k \leq n$.

\subsubsection{Encoding} 
%
Consider a source sequences of length $Bn$ $s^{Bn}_{i} \in \mS^{Bn}_i, i=1,2$. Partition each sequence into $B$ length-$n$ subsequences, $\svec_{i,b}, b=1,\dots,B$. Similarly, for $b=1,2,\dots,B$, partition the side information sequences $w_3^{Bn}$ and $w^{Bn}$ into $B$ length-$n$ subsequences $\wvec_{3,b}, \wvec_{b}$, respectively. We transmit a total of $Bn$ source samples over $B+1$ blocks of $n$ channel uses each.

At block $1$, source terminal $i, i=1,2$, observes $\svec_{i,1}$ and finds its corresponding bin index $\msg_{i,1} \in \msgCal_i$. It then transmits the channel codeword $\xvec_i(\msg_{i,1}, \avec_i)$ where $\avec_i \in \mS_i^n$ is a fixed sequence.
At block $b, b=2,\dots,B$, source terminal $i, i=1,2$, transmits the channel codeword $\xvec_i(\msg_{i,b},\svec_{i,b-1})$ where $\msg_{i,b} \in \msgCal_i$ is the bin index of source vector $\svec_{i,b}$.
At block $B+1$, source terminal $i, i=1,2$, transmits $\xvec_i(1,\svec_{i,B})$.

At block $b=1$, the relay transmits $\xvec_3(\avec_1,\avec_2)$.
Assume that at block $b, b=2,\dots,B,B+1$, the relay obtained the estimates $(\tsvec_{1,b-1},\tsvec_{2,b-1})$ of $(\svec_{1,b-1},\svec_{2,b-1})$. It then transmits the channel codeword $\xvec_3(\tsvec_{1,b-1},\tsvec_{2,b-1})$.

\subsubsection{Decoding}
The relay decodes the source sequences sequentially trying to reconstruct source block  $\svec_{i,b}, i=1,2$, at the end of channel block $b$ as follows: let $\tsvec_{i,b-1}, i=1,2,$ be the estimate of $\svec_{i,b-1}$, at the relay at the end of block $b-1$. Using this information, and its received signal $\yvec_{3,b}$, the relay channel decoder at time $b$ decodes $(\msg_{1,b}, \msg_{2,b})$, i.e., the bin indices corresponding to $\svec_{i,b}, i=1,2$, by looking for a unique pair $(\tmsg_{1}, \tmsg_{2})$ such that:
\vspace{-0.25cm}
\begin{align}
    & \big(\tsvec_{1,b-1}, \tsvec_{2,b-1}, \xvec_1(\tilde{\msg}_{1}, \tsvec_{1,b-1}), \xvec_2(\tilde{\msg}_{2}, \tsvec_{2,b-1}), \nonumber \\
    & \qquad \xvec_3(\tsvec_{1,b-1}, \tsvec_{2,b-1}), \yvec_{3,b}\big) \in \styp.
    \label{eq:RelayJntFlipDecType}
\end{align}
\vspace{-0.6cm}

The decoded bin indices, denoted $(\tilde{\msg}_{1,b}, \tilde{\msg}_{2,b})$, are then given to the relay source decoder.
Using $(\tilde{\msg}_{1,b}, \tilde{\msg}_{2,b})$ and the side information $\wvec_{3,b}$, the relay source decoder estimates
$(\svec_{1,b}, \svec_{2,b})$ by looking for a unique pair of sequences $(\tsvec_{1}, \tsvec_{2})$ that satisfies $f_1(\tilde{\svec}_{1})= \tilde{\msg}_{1,b}, f_2(\tilde{\svec}_{2})= \tilde{\msg}_{2,b}$ and $(\tilde{\svec}_{1},\tilde{\svec}_{2},\wvec_{3,b}) \in \styp(S_1,S_2,W_3)$. Let $(\tsvec_{1,b}, \tsvec_{2,b})$ denote the decoded sequences.

Decoding at the destination is done using backward decoding. The destination node waits until the end of channel block $B+1$.
It first decodes $\svec_{i,B}, i=1,2$, using the received signal at channel block $B+1$ and its side information $\wvec_{B}$.
Going backwards from the last channel block to the first, we assume that the destination has estimates $(\hsvec_{1,b+1},\hsvec_{2,b+1})$ of $(\svec_{1,b+1},\svec_{2,b+1})$ and consider decoding of $(\svec_{1,b},\svec_{2,b})$. From $(\hsvec_{1,b+1}, \hsvec_{2,b+1})$ the destination finds the corresponding bin indices $(\hmsg_{1,b+1}, \hmsg_{2,b+1})$. Using this information, its received signal $\yvec_{b+1}$ and the side information $\wvec_{b}$, the destination decodes $(\svec_{1,b},\svec_{2,b})$ by looking for a unique pair $(\hsvec_{1}, \hsvec_{2})$ such that:
\vspace{-0.3cm}
\begin{align}
    & \big(\hsvec_{1}, \hsvec_{2}, \xvec_1(\hmsg_{1,b+1}, \hsvec_{1}), \xvec_2(\hmsg_{2,b+1},\hsvec_{2}), \nonumber \\
    & \qquad \qquad \qquad \xvec_3(\hsvec_{1}, \hsvec_{2}), \wvec_{b}, \yvec_{b+1}\big) \in \styp.
\label{eq:DestFlipJntChanDecType}
\end{align}
\vspace{-0.6cm}

\subsubsection{Error probability analysis} The error probability analysis is detailed in \cite[Appendix D]{Murin:IT11}. $\qedhere$
\end{proof} 


\vspace{-0.1cm}
\section{Discussion} \label{sec:jointDiscussion}
\vspace{-0.1cm}

%

\begin{remark}
    Constraints \eqref{bnd:Joint_rly_S1}--\eqref{bnd:Joint_rly_S1S2} in \Thmref{thm:jointCond} and \eqref{bnd:JointFlip_rly_S1}--\eqref{bnd:JointFlip_rly_S1S2} in \Thmref{thm:jointCondFlip}, are due to decoding at the relay, while constraints \eqref{bnd:Joint_dst_S1}--\eqref{bnd:Joint_dst_S1S2} in \Thmref{thm:jointCond} and \eqref{bnd:JointFlip_dst_S1}--\eqref{bnd:JointFlip_dst_S1S2} in \Thmref{thm:jointCondFlip}, are due to decoding at the destination.
\end{remark}

\begin{remark}
	In \Thmref{thm:jointCond} $V_1$ and $V_2$ represent the binning information for $S_1$ and $S_2$, respectively. Observe that the left-hand side (LHS) of condition \eqref{bnd:Joint_rly_S1} is the entropy of $S_1$ when $(S_2,W_3)$ are known. 
On the right-hand side (RHS) of \eqref{bnd:Joint_rly_S1}, as $V_1$, $S_2$, $X_2$, $X_3$ and $W_3$ are given, the mutual information expression $I(X_1;Y_3|S_2, X_2, V_1, X_3, W_3)$ represents the available rate that can be used for sending to the relay information on the {\em source sequence} $S_1^n$, in excess of the bin index represented by $V_1$.
The LHS of condition \eqref{bnd:Joint_dst_S1} is the entropy of $S_1$ when $(S_2,W)$ are known. The RHS of condition \eqref{bnd:Joint_dst_S1} expresses the rate at which {\em binning information} can be transmitted reliably and cooperatively from transmitter 1 and the relay to the destination.
 This follows as the mutual information expression on the RHS of \eqref{bnd:Joint_dst_S1} can be written as $I(X_1,X_3;Y|S_1, X_2, V_2) = I(X_1,X_3;Y|S_1, S_2, V_2, X_2, W)$, which, as $S_1$, $S_2$ and $W$ are given, represents the rate for sending the {\em bin index} of source sequence $S_1^n$ to the destination (see \cite[Subsection VI-D]{Murin:IT11}).
This is in contrast to the decoding constraint at the relay, c.f. \eqref{bnd:Joint_rly_S1}.
Therefore, each mutual information expression in \eqref{bnd:Joint_rly_S1} and \eqref{bnd:Joint_dst_S1} represents a \emph{different} type of information sent by the source: either the source-channel codeword to the relay in \eqref{bnd:Joint_rly_S1}, or bin index to the destination in \eqref{bnd:Joint_dst_S1}.
This is because SSCC is used for sending information to the destination and CPM is used for sending information to the relay.

  In \Thmref{thm:jointCondFlip} the LHS of condition \eqref{bnd:JointFlip_rly_S1} is the entropy of $S_1$ when $(S_2,W_3)$ are known. 
In the RHS of \eqref{bnd:JointFlip_rly_S1} the mutual information expression $I(X_1;Y_3|S_1, X_2, X_3) = I(X_1;Y_3|S_1, S_2, X_2, X_3, W_3)$ represents the rate for sending the {\em bin index} of the source sequence $S_1^n$ to the relay (see \cite[Subsection VI-E]{Murin:IT11}). This is because $S_1$, $S_2$ and $W_3$ are given.
The LHS of condition \eqref{bnd:JointFlip_dst_S1} is the entropy of $S_1$ when $(S_2,W)$ are known. 
In the RHS of condition \eqref{bnd:JointFlip_dst_S1}, as $S_2$, $X_2$ and $W$ are given, the mutual information expression $I(X_1,X_3;Y|S_2, X_2,W)$ represents the available rate that can be used for sending information on the {\em source sequence} $S_1^n$ to the destination.
\end{remark}

\vspace{-0.15cm}

\begin{remark}
    For an input distribution
    \vspace{-0.2cm}
    \begin{align*}
        & p(s_1,s_2,w_3,w,v_1,v_2,x_1,x_2,x_3) = \nonumber \\
        & \qquad p(s_1,s_2,w_3,w)p(v_1)p(x_1|v_1)p(v_2)p(x_2|v_2)p(x_3|v_1,v_2),
    \end{align*}
    
    \vspace{-0.20cm}
    
     \noindent the conditions in \eqref{bnd:Joint} specialize to \cite[Equation (2)]{Murin:ISWCS11}, and the transmission scheme specializes to a separation-based achievability scheme.
\end{remark}

\begin{remark} \label{rem:JointReduce}
    In both \Thmref{thm:jointCond} and \Thmref{thm:jointCondFlip} the conditions stemming from the CPM technique can be specialized to the MAC source-channel conditions of \cite[Equations (12)]{Cover:80}.
    In \Thmref{thm:jointCond} letting $\mV_1= \mV_2= \mX_3= \mW_3=\phi$, reduces the relay conditions in \eqref{bnd:Joint_rly_S1}--\eqref{bnd:Joint_rly_S1S2} to the ones in \cite[Equations (12)]{Cover:80} with $Y_3$ as the destination.
    In \Thmref{thm:jointCondFlip} letting $\mX_3=\mW=\phi$, reduces the destination conditions in \eqref{bnd:Joint_dst_S1}--\eqref{bnd:Joint_dst_S1S2} to the ones in \cite[Equations (12)]{Cover:80} with $Y$ as the destination.
\end{remark}

\begin{remark} \label{rem:tradeoff}
    \Thmref{thm:jointCond} and \Thmref{thm:jointCondFlip} establish \emph{different} achievability conditions.
    As stated in Section \ref{subsec:exampleSepSubOpt}, SSCC is generally suboptimal for sending correlated sources over DM MARCs and MABRCs.
    In \Thmref{thm:jointCond} the CPM technique is used for sending information to the relay, while in \Thmref{thm:jointCondFlip} SSCC is used for sending information to the relay. This observation implies that the relay decoding constraints of \Thmref{thm:jointCond} are looser compared to the relay decoding constraints of \Thmref{thm:jointCondFlip}.
    Using similar reasoning we conclude that the destination decoding constraints of \Thmref{thm:jointCondFlip} are looser compared to the destination decoding constraints of \Thmref{thm:jointCond} (as long as coordination is possible, see Remark \ref{rem:CRBCrem}).
    Considering the distribution chains in \eqref{eq:JntJointDist} and \eqref{eq:JntFlipJointDist} we conclude that these two theorems represent different sets of sufficient conditions, and neither theorem is a special case of the other.
\end{remark}

\begin{remark} \label{rem:CRBCrem}
Figure \ref{fig:CBRCsideInfo} depicts the cooperative relay broadcast channel (CRBC) model with correlated relay and destination side information, which is a special case of the MABRC with $\mX_2=\mS_2=\phi$. For this model the optimal source-channel rate was obtained in \cite[Theorem 3.1]{ErkipGunduz:07}:

\begin{figure}[h]
    \vspace{-0.2cm}
    \centering
    \scalebox{0.4}{\includegraphics{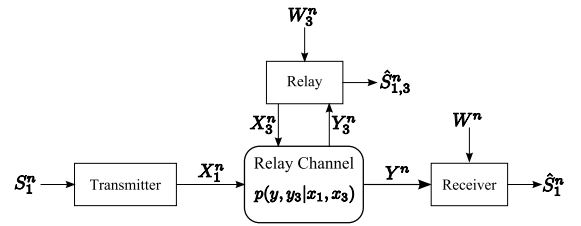}}
    \vspace{-0.1cm}
    \caption{Cooperative relay broadcast channel with correlated side information.}
    \label{fig:CBRCsideInfo}
    \vspace{-0.4cm}
\end{figure}

\begin{propositionA}[{\cite[Theorem 3.1]{ErkipGunduz:07}}] \label{cor:relayJoint}
   A source $S_1^n$ can be reliably transmitted over a DM CRBC with relay and receiver side information if
   \vspace{-0.33cm}
    \begin{subequations} \label{eq:relayJoint}
        \begin{eqnarray}
            H(S_1|W_3) &<& I(X_1;Y_3|X_3) \label{eq:RelayJnt_cond} \\
            H(S_1|W) &<& I(X_1,X_3;Y), \label{eq:DestJnt_cond}
        \end{eqnarray}
    \end{subequations}
   	\vspace{-0.65cm}
   
    \noindent for some input distribution $p(s_1, w_3, w)p(x_1, x_3)$. Conversely, if 
		a source $S_1^n$ can be reliably transmitted then the conditions in \eqref{eq:RelayJnt_cond} and \eqref{eq:DestJnt_cond} are satisfied with $<$ replaced by $\leq$ for some input distribution $p(x_1, x_3)$.
\end{propositionA}


The conditions in \eqref{eq:relayJoint} can also be obtained from \Thmref{thm:jointCond} by letting $V_1=X_3$, $\mathcal{S}_2 =\mathcal{X}_2 =\mathcal{V}_2=\phi$, and considering an input distribution independent of the sources. However, when we consider \Thmref{thm:jointCondFlip} with $\mathcal{S}_2=\mathcal{X}_2=\phi$, we obtain the following achievability conditions:
	\vspace{-0.2cm}
  \begin{subequations} \label{eq:relayFlipJoint}
    \begin{eqnarray}
        H(S_1|W_3) &<& I(X_1;Y_3|X_3,S_1) \label{eq:RelayFlipJnt_cond} \\
        H(S_1|W) &<& I(X_1,X_3;Y|W), \label{eq:DestFlipJnt_cond}
    \end{eqnarray}
  \end{subequations}
  \vspace{-0.63cm}
  
\noindent for some input distribution $p(s_1, w_3, w)p(x_1|s_1)p(x_3|s_1)$.

Note that the RHS of the inequalities in \eqref{eq:RelayFlipJnt_cond} and \eqref{eq:DestFlipJnt_cond} are not greater than the RHS of the inequalities in \eqref{eq:RelayJnt_cond} and \eqref{eq:DestJnt_cond}, respectively. Moreover, not all joint input distributions $p(x_1,x_3)$ can be achieved via $p(x_1|s_1)p(x_3|s_1)$. Hence, the conditions obtained from \Thmref{thm:jointCondFlip} for the CRBC setup are stricter than those obtained from \Thmref{thm:jointCond}, illustrating the fact that the two sets of conditions are not equivalent.
We conclude that the downside of using CPM to the destination as applied in this work is that it puts constraints on the distribution chain, thereby constraining the achievable coordination between the sources and the relay. For this reason, when there is only a single source, the joint distributions of the source and the relay ($X_1$ and $X_3$) achieved by the scheme of \Thmref{thm:jointCondFlip}, do not exhaust the entire space of joint distributions, resulting in generally stricter source-channel constraints than those obtained from \Thmref{thm:jointCond}.
However, recall that for SOMARC in Section \ref{subsec:exampleSepSubOpt} the optimal scheme uses CPM to the destination. Therefore, for the MARC it is not possible to determine whether either of the schemes is universally better than the other.
\end{remark}


\begin{remark}
    In both \Thmref{thm:jointCond} and \Thmref{thm:jointCondFlip} we use a combination of SSCC and CPM. Since CPM can generally support sources with higher entropies, a natural question that arises is whether it is possible
    to design a scheme based only on CPM; namely, encode both cooperation (relay) information and the new information, using a superposition CPM scheme. 
		This approach cannot be used directly in the framework of the current paper. Here, we use joint typicality decoder, which does not apply to different blocks generated independently with the same distribution.
    For example, we cannot test the joint typicality of $s_{1,1}^n$ and $s_{1,n+1}^{2n}$, as they belong to different time blocks. Using a CPM-only scheme would require such a test.
 We conclude that applying the CPM technique for sending information to both the relay and the destination cannot be done while using joint typicality decoder as considered in this paper. It is, of course, possible
 to construct schemes that use a different decoder, or apply CPM through intermediate RVs, which overcome this difficulty. Investigation of such coding schemes is left for future research.
\end{remark}


\vspace{-0.2cm}

\section{Conclusions} \label{sec:conclusions}

\vspace{-0.1cm}

In this paper we considered joint source-channel coding for DM MARCs and MABRCs.
We first showed via an explicit example that joint source-channel coding generally enlarges the set of possible sources that can be reliably transmitted compared to separation-based coding. We then derived two new joint source-channel achievability schemes.
Both schemes use a combination of SSCC and CPM techniques.
While in the first scheme CPM is used for encoding information to the relay and SSCC is used for encoding information to the destination, in the second scheme SSCC is used for encoding information to the relay and CPM is used for encoding information to the destination. The different combinations of binning and source mapping enable flexibility in the system design by choosing one of the two schemes according to the quality of the side information and received signals at the relay and at the destination. In particular, the first scheme has looser decoding constraints at the relay and is therefore better when the channels from the sources to the relay are the bottleneck; while the second scheme has looser decoding constraints at the destination, and is more suitable for scenarios in which the channels to the destination are more noisy (at the cost of more constrained source-relay coordination).

\end{document}